\documentclass[preprint,12pt]{elsarticle}
\usepackage[T1]{fontenc}
\usepackage[utf8]{inputenc}

\usepackage{algorithmicx}
\usepackage[ruled]{algorithm}
\usepackage{algpseudocode}
\usepackage{pgfplots}
\pgfplotsset{compat=newest}

\usepackage{hyperref}
\hypersetup{
	colorlinks=true,
	linkcolor=black,     
	urlcolor=blue,
	citecolor=black,
}

\usepackage{amssymb}

\usepackage{amsmath,mathtools,amsthm}

\usepackage{graphicx,subcaption,caption}

\DeclareMathOperator{\sgn}{sgn}
\DeclareMathOperator{\vect}{vec}

\DeclareMathOperator*{\argmin}{arg\,min}

\DeclarePairedDelimiter\norm{\lVert}{\rVert}

\makeatletter
\let\oldnorm\norm
\def\norm{\@ifstar{\oldnorm}{\oldnorm*}}
\makeatother

\newtheorem{proposition}{Proposition}

\newcommand{\review}[1]{\textcolor{black}{#1}}

\journal{Signal Processing}

\begin{document}

\begin{frontmatter}



\title{Through the Wall Radar Imaging via Kronecker-structured \review{Huber-type} RPCA}

\affiliation[sondra]{organization={SONDRA, CentraleSupélec},
            city={Gif-sur-Yvette},
            postcode={91190}, 
            country={France}}
        
\affiliation[leme]{organization={LEME, Université Paris-Nanterre},
	city={Ville d'Avray},
	postcode={92410}, 
	country={France}}

\affiliation[listic]{organization={LISTIC, Université Savoie Mont-Blanc},
	city={Annecy},
	postcode={74940}, 
	country={France}}

\author[sondra]{Hugo Brehier\corref{cor1}}
\cortext[cor1]{Corresponding author}
\ead{hugo.brehier@centralesupelec.fr}

\author[leme]{Arnaud Breloy}
\author[sondra]{Chengfang Ren}
\author[listic]{Guillaume Ginolhac}

\begin{abstract}
The detection of multiple targets in an enclosed scene, from its outside, is a challenging topic of research addressed by Through-the-Wall Radar Imaging (TWRI). Traditionally, TWRI methods operate in two steps: first the removal of wall clutter then followed by the recovery of targets positions. Recent approaches manage in parallel the processing of the wall and targets via low rank plus sparse matrix decomposition and obtain better performances. In this paper, we reformulate this precisely via a RPCA-type problem, where the sparse vector appears in a Kronecker product. We extend this approach by adding a robust distance with flexible structure to handle heterogeneous noise and outliers, which may appear in TWRI measurements. The resolution is achieved via the Alternating Direction Method of Multipliers (ADMM) and variable splitting to decouple the constraints. The removal of the front wall is achieved via a closed-form proximal evaluation and the recovery of targets is possible via a tailored Majorization-Minimization (MM) step. The analysis and validation of our method is carried out using Finite-Difference Time-Domain (FDTD) simulated data, which show the advantage of our method in detection performance over complex scenarios.
\end{abstract}


\begin{highlights}
\item Reformulation of TWRI detection via RPCA
\item Addition of a structured robust distance for heterogeneous noises
\item Resolution via ADMM coupled with variable splitting
\item Analysis of the proposed the method on FDTD simulations
\end{highlights}


\end{frontmatter}


\section{Introduction}
\label{sec:intro}

Through the Wall Radar Imaging is a current topic of research (see e.g. \cite{amin_TTW_book} for a comprehensive review) that aims at detecting targets in an enclosed scene from its outside via radar measurements, the scene being unobservable to the naked eye. 
It makes use of the penetrative properties of electromagnetic waves to obtain returns from the inside of the scene while having to filter the front wall echoes. The ability to observe scenes though wall or other similar types of obstacle would be a useful technique for military operations, civilian rescue operations and monitoring \cite{Li2021_ttw_monitoring,Yang2021_monitoring}. 
Many problems appear in the context of TWRI. 
Firstly, the front wall (facing the radar) returns are overwhelming and obscure the enclosed scene. 
Secondly, the echoes from the enclosed scene are subject to different phenomena: clutter from the inner walls gets mixed with the target echoes.
Moreover, those returns can travel across different paths, so-called multipaths, which may create ghost targets. 

Past works have focused on different aspects of the topic of TWRI such as localisation of targets, change detection, movement characterization \cite{debes2011_classif, clemente2013microdoppler, Gennarelli2015_tracking, Li2019_tracking}. Here, we focus on the localisation of stationary targets which can be readily extended to moving targets by collecting measurement over time and applying the same methodology. We will focus on a 2D scenario which necessitates the use of multiple antennas (or a single travelling one) to achieve a sufficient resolution.
A standard hypothesis in TWRI is for the wall to be homogeneous, with permittivity and thickness considered to be known, or to be estimated in a previous step \cite{protiva2011_estim_wall_params, Jin2013_estim_wall_params}. Other works have developed methods for the unknown case based on focusing techniques \cite{amin2006, ahmad_2007}.
In an earlier phase of TWRI, some methods \cite{ahmad_2008, Dehmollaian2008_sar} were developed that use Synthetic Aperture Radar (SAR) techniques \cite{Soumekh_SAR} such as Back-Projection (BP). Those methods require the acquisition of measurements from an empty scene to remove the front wall. 
Subsequently, two-step techniques were developed \cite{CS_TTW_Radar_Imaging} which consist in: a) filtering the front wall echoes based on subspace decomposition \cite{verma2009clutter,svd_clutter_mitig, subspace_proj} b) recovering the target positions, based on the hypothesis of sparsity of the targets w.r.t. the scene dimensions, with the possible use of Compressive Sensing methods to reduce computation times \cite{huang2010}. 
This approach requires the use of a dictionary to map the returns onto a grid covering the scene. This formalism also allows handling multipaths or front wall reflections more precisely \cite{leigsnering2014multipath}. 
Building on this, one-step methods have been explored during the past years via the framework of Robust Principal Component Analysis (RPCA) \cite{candes2011rpca,chandrasekaran2011rpca,Mardani_2013} which allows the joint decomposition of a matrix in two separate components: one being low rank and the other being sparse, the two parts capturing respectively the returns of the front wall and the returns of the targets. Such one-step methods have been shown to perform better than their counterparts in several radar experiments \cite{tang2016, tang2020, breloy2018_rsc_rd, meriaux2019_modifssc, brehier2022robust}.

We build upon these more recent approaches and address some of their limitations in the context of TWRI. A question to be raised is the robustness of those methods in the case that the measurements do not respect the model perfectly. Indeed, the returns from the front wall may not be homogeneous as supposed: the structure of drywall may for example create a discrepancy of the returns power among the different radar positions. Moreover, the permittivity of the wall that is supposed to be frequency-independent may be too restrictive and its treatment may lead to better performances. This has motivated us to inspect the addition of a robust distance \cite{huber1964} with flexible structure to one-step matrix decomposition methods applied to TWRI.

To do so, we first formalize the TWRI problem in the context of RPCA, which leads to the sparse component appearing in a Kronecker product, a special case of the model in \cite{Mardani_2013}. We make adjustments to previous work in TWRI by using the ADMM framework \cite{admm}. This first part concluded with the presentation of a method, we then introduce the use of a robust distance with a flexible structure. This allows us to handle heterogeneous noise or outliers in the data closeness term, which may distort the results grossly with the usual euclidian distance. This was suggested in \cite{Aravkin2014_varpcp} but not developed. We present two methods for the resolution of this problem which make use of ADMM in tandem with variable splitting and either Proximal Gradient Descent (PGD) or MM frameworks.

The following sections of the paper are organized as follows. Section \ref{sec:existing_model} presents a standard model for the measurements and describes existing two-step and one-step  methods for TWRI. In Section \ref{sec:robust_krpca}, we introduce and develop the robust extension to one-step methods. Section \ref{sec:experiments} follows with experiments done on simulated data to compare the performance of the different methods. Finally, Section \ref{sec:conclusion} summarizes the advantages of our method and perspectives of future work.

\section{Existing TWRI models and methods} \label{sec:existing_model}
\subsection{Setting and signal model}

\begin{figure}[!h]
	\centering
	\includegraphics[width = .5\textwidth]{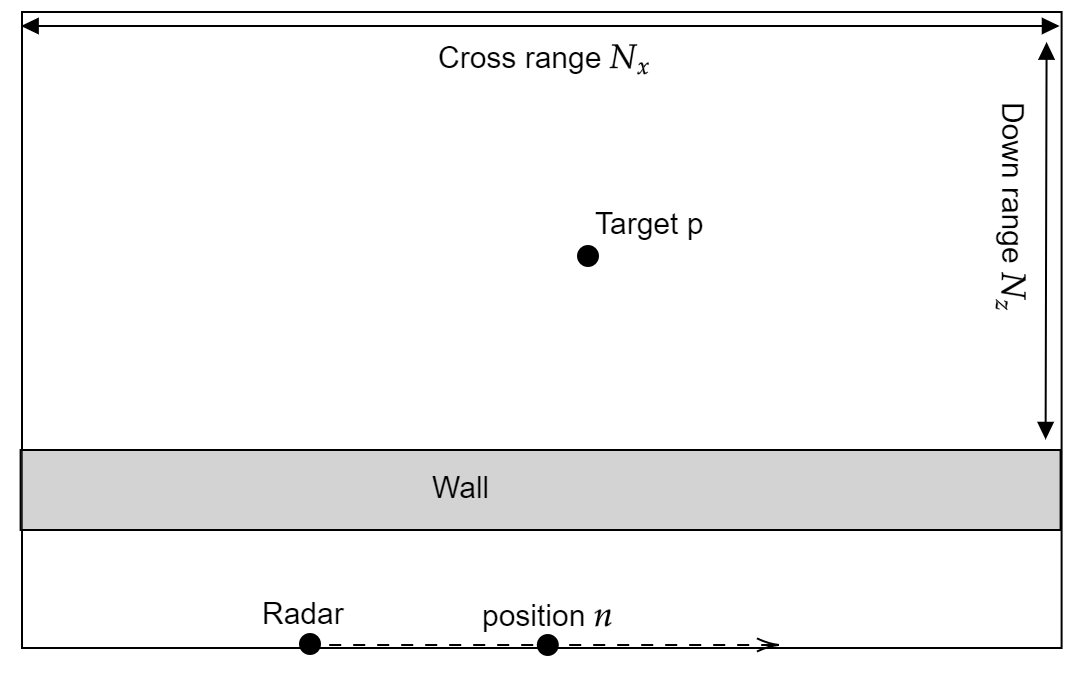}
	\caption{2D Through the Wall setting (view from above)}
	\label{fig:scene_desc}
\end{figure}

\begin{figure}[!h]
	\includegraphics[width = .48\linewidth]{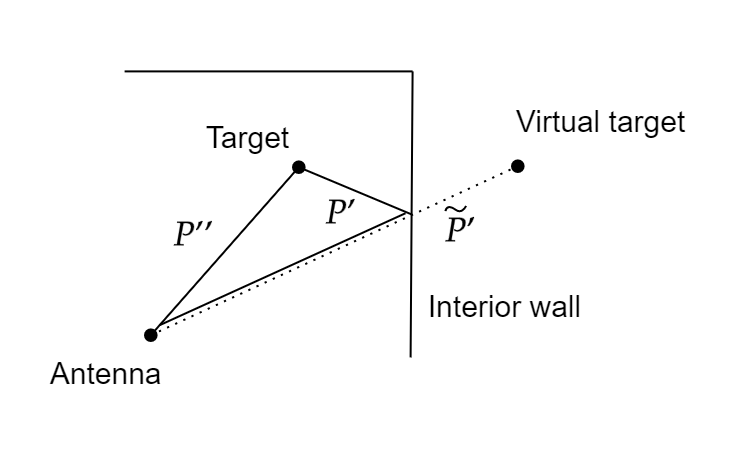}
	\includegraphics[width = .48\linewidth]{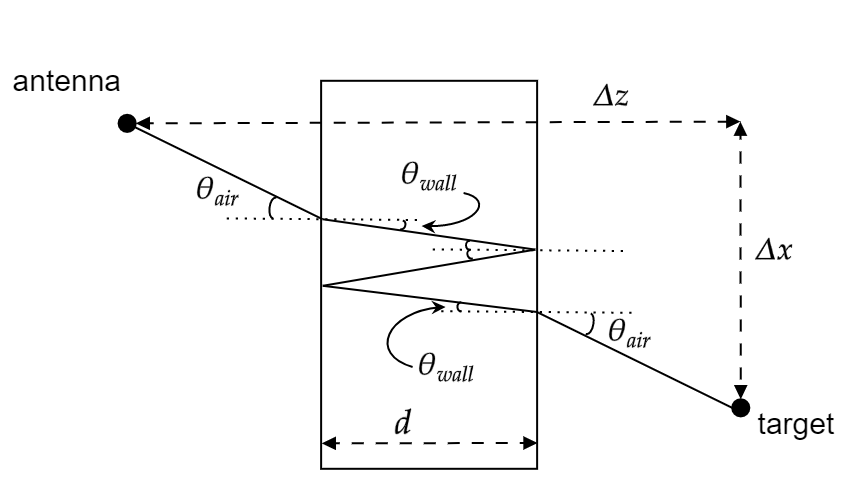}
	\caption{Multipath propagation via reflection at an interior wall (left) and via reverberation i.e. "wall ringing" (right)}
	\label{fig:examples_multipaths}
\end{figure}

We first present a signal model widely used in the TWRI literature \cite{CS_TTW_Radar_Imaging,tang2020}. Consider a 2D scene, as in Figure \ref{fig:scene_desc}, with homogeneous wall of thickness $d$ and permittivity $\epsilon$ located along the $x$-axis. We obtain $N$ measurements over a synthetic array parallel to the wall at a standoff distance $z_{off}$, with a stepped-frequency radar signal of $M$ frequencies uniformly spaced over some frequency band $[\omega_1, \omega_{M}]$ so that:  $\omega_m = \omega_1 + m \Delta \omega , \quad m = 1,\ldots, M$ with $\Delta\omega$ the frequency step. Measurements are in a stop-and-go setting. 

The noiseless received signal can be written as the superposition of the returns from the front-wall and targets. This leaves out from the model echoes stemming from the inner walls which may form clutter in real-world scenarios.
\review{
For the $m^{th}$ frequency and $n^{th}$ position, we have:
\begin{equation}\label{eq:received_signal_mp}
		y(m,n) = \sum_{k=1}^{K} \sigma_w^{(k)} \exp{(-j \omega_m \tau_w^{(k)})} 
		+ \sum_{i=1}^{R} \sum_{p=1}^{P} \sigma_p^{(i)} \exp(-j \omega_m \tau_{p,n}^{(i)})
\end{equation}
where $P$ is the number of point targets in the scene, considered to be low w.r.t. the scene dimensions, $K$ is the number of reverberations in the front wall, while $R$ is the number of possible multipaths, as the propagation through the front wall and the lateral/back walls induce several possible paths as presented in Figure \ref{fig:examples_multipaths}.
}

\review{
Moreover, $\sigma_w^{(k)}$ is a complex-valued attenuation coefficient comprising the reflectivity of the wall and the path loss for the $k^{th}$ wall reverberation, $\tau_w^{(k)}$ is the round trip propagation delay from transceiver to wall. }

Additionally,  $\sigma_p^{(i)}$ is a complex-valued attenuation coefficient which factors in the different losses for the $i^{th}$ multipath to the $p^{th}$ target: the wall refraction loss, path loss in air and wall, and target reflection loss. The two-way propagation delay from $n^{th}$ transceiver to the $p^{th}$ target along the $i^{th}$ multipath is denoted $ \tau_{p,n}^{(i)}$. The direct trajectory through the front wall can be computed by numerical methods as in \cite{ahmad_2006}, which allows us to evaluate the associated propagation delay.

For numerical evaluation and implementation, we discretize the scene into a grid of dimension $N_x \times  N_z$ in crossrange vs downrange. We now denote $ \tau^{(i)}_{n_x n_z,n}$ the propagation delay to the $(n_x,n_z)^{th}$ pixel for the $i^{th}$ multipath scheme and the $n^{th}$ transceiver position. Then, we can write the received signal through a dictionary $\mathbf{\Psi}$ which maps the whole scene. For the $i^{th}$ multipath scheme and the $n^{th}$ transceiver position, its $(n_x,n_z)^{th}$ column describes the return from a point target at the $(n_x,n_z)^{th}$ pixel:
\begin{equation}
	[\mathbf{\Psi}_n^{(i)}]_{n_x n_z} = [\exp{(-j \omega_0 \tau^{(i)}_{n_x n_z,n})} \ldots \exp{(-j \omega_{M-1} \tau^{(i)}_{n_x n_z,n})} ]^T
\end{equation}
Then, $\mathbf{\Psi}_n \in \mathbb{C}^{M \times N_x N_z R}$ is the part of the overall dictionary mapping from received signal to target positions from the $n^{th}$ transceiver position. This allows to write in vector form  the signal received at the $n^{th}$ position:
\begin{equation}
	\begin{split}
		\mathbf{y}_n &= 
		\mathbf{l}_n + 
		\underbrace{[ \mathbf{\Psi}_n^{(1)} \mathbf{\Psi}_n^{(2)} \ldots \mathbf{\Psi}_n^{(R)}]}_{= \mathbf{\Psi}_n}
		\underbrace{\begin{bmatrix}
				\mathbf{r}^{(1)} \\
				\mathbf{r}^{(2)}\\
				\vdots \\
				\mathbf{r}^{(R)}
		\end{bmatrix}}_{=\mathbf{r}} \\
		\implies \mathbf{y}_n  &=  \mathbf{l}_n  +  \mathbf{\Psi}_n \mathbf{r}
	\end{split}
\end{equation}
where $\mathbf{l}_n  \in \mathbb{C}^{M}$ contains the returns of the front wall and $\mathbf{r}^{(i)} \in \mathbb{C}^{N_x N_z}$ is the scene vector associated to the $i^{th}$ multipath propagation scheme containing the back-scattered signal complex amplitudes over the grid covering the scene.

In this setting, the returns of the front wall are overwhelming. It is needed to filter them out in order to achieve the detection of targets. The most convenient and common technique is to separate the target and wall subspaces to mitigate the contribution to the measurements of the wall. 
Indeed, the radar displacement axis being chosen to be parallel to the front wall induces an invariance of the front wall returns along the different measurement positions. Coupled with their higher power compared with the scene behind, this calls for the use of a subspace decomposition.
The filtering of the front wall followed by the detection of targets can be stated as a two-step method which we detail in the next section.

\subsection{SR-CS: vectorized overall model and two-step methods}

The method of \citet{CS_TTW_Radar_Imaging}, which we denote SR-CS (for Sparse Recovery - Compressed Sensing) considers a vectorized model where the total signal model is created by stacking the measurements at the $N$ radar positions in a long composite vector. 

Indeed, let $\mathbf{y} \triangleq [\mathbf{y}_{1}^T \ldots \mathbf{y}_{N}^T]^T$ and $\mathbf{\Psi}_{A} \triangleq [\mathbf{\Psi}_1^T \ldots \mathbf{\Psi}_{N}^T]^T$ so that: 
\begin{equation}
\mathbf{y}  = [\mathbf{l}_1^T,\ldots,\mathbf{l}_{N}^T]^T + \mathbf{\Psi}_{A} \mathbf{r}
\end{equation}
SR-CS assumes that the front wall returns have been suppressed (see e.g. \cite{svd_clutter_mitig,subspace_proj}) so that $\mathbf{l}_n = \mathbf{0} \; \forall n=1, 2\ldots, N$ and $\mathbf{y} = \mathbf{\Psi}_{A} \mathbf{r}$.


Assuming that the number of targets is relatively low w.r.t. the scene dimensions, the vector of amplitudes $\mathbf{r}$ is sparse. The recovery of $\mathbf{r}$  with a sparsity regularization consists in a renowned problem of sparse recovery, the most famous example being the LASSO regression \cite{Tibshirani1996_lasso}, which uses a $\ell_1$ norm regularization. The use of the $\ell_{2,1}$ norm, a regularization used in order to promote grouped sparsity across rows \cite{kowalski_mixed_norms}, has been developed for multipath exploitation in \cite{leigsnering2014multipath}. It is defined as the sum of the euclidian norm of the rows of a matrix. Indeed, in our case, rows represent one pixel viewed across different multipaths, we then want our method to promote activation of whole rows,  as the underlying scene is the same across multipaths. This can filter out multipath ghosts which appear at some position in an unstructured way, i.e. not across all multipaths, on the contrary of true targets. 

\subsection{KRPCA: matricized overall model and one-step methods}
\subsubsection{Data Model}
The approach in Section \ref{sec:existing_model} is a sequential method in two steps : \textit{a}) filter the front wall, \textit{b}) recover the target positions. Recent works \cite{tang2020, brehier2022robust} suggest that a parallel recovery of both components can improve performances. 

This can be considered through a decomposition of the data matrix, more precisely low rank plus sparse decomposition methods. This was notably developed in the framework of Robust PCA (RPCA) \cite{candes2011rpca,chandrasekaran2011rpca} whose goal is to retrieve a low-dimensional subspace in which lie the data points, except for some outliers which are accounted for in a sparse matrix. It makes use of the $\ell_1$ and nuclear norms for convex relaxation, known to be the convex envelopes of the $\ell_0$ `norm' (the number of non-zeros entries) and rank of a (bounded) matrix \cite{fazel2002matrix}.
In \cite{Mardani_2013}, it was extended to a setting with a compressing operator acting on the sparse component. 
However, we may observe that our model is a special case of the aforementioned method. Indeed, note that we can write an overall matricized model for the observations:
\begin{equation} \label{eq:decomp_signal_model_mp}
	\begin{split}
		\underbrace{\left[ \mathbf{y}_1 \ldots \mathbf{y}_{N} \right]}_{= \mathbf{Y}} &=
		\underbrace{\left[ \mathbf{l}_1 \ldots \mathbf{l}_{N} \right]}_{= \mathbf{L}} +
		\underbrace{ \left[ \mathbf{\Psi}_1 \ldots \mathbf{\Psi}_{N} \right]}_{= \mathbf{\Psi}}
		\underbrace{ \begin{bmatrix}
				\mathbf{r} &  \mathbf{0}  & \ldots & \mathbf{0} \\ 
				\mathbf{0} & \ddots & \ddots  &\vdots \\
				\vdots & \ddots  & \ddots & \mathbf{0} \\ 
				\mathbf{0} & \ldots &  \mathbf{0} & \mathbf{r}
		\end{bmatrix}}_{= \mathbf{I}_N \otimes \mathbf{r}} \\
		\implies \mathbf{Y} &= \mathbf{L} + \mathbf{\Psi} \left( \mathbf{I}_N \otimes  \mathbf{r} \right)
	\end{split}
\end{equation}
with $\otimes$ denoting the Kronecker product. $\mathbf{Y} \in \mathbb{C}^{M \times N}$ the data matrix, $\mathbf{L} \in \mathbb{C}^{M \times N}$ a low-rank matrix of front wall returns, $\mathbf{\Psi} \in \mathbb{C}^{M \times N_x N_z R N }$ a dictionary mapping to the target returns and $\mathbf{S} \in \mathbb{C}^{N_x N_z R N \times N }$ the associated sparse matrix containing the scene vector.

\subsubsection{Problem statement}

Some works of low-rank plus sparse matrix decomposition exist in the context of TWRI \cite{tang2016, tang2020}. Additionally, the work of \cite{brehier2022robust}, denoted KRPCA (for Kronecker-structured RPCA), proposed the following formulation to refine the model of \cite{Mardani_2013}:
\begin{equation} \label{eq:modified_dico_rpca}
	\begin{split}
		\min_{\mathbf{L},\mathbf{r}} &\quad \norm{\mathbf{L}}_* + \lambda \norm{ \vect^{-1} (\mathbf{r})}_{2,1} \\
		\text{s.t.} &\quad  \mathbf{Y} = \mathbf{L} +  \mathbf{\Psi} \left( \mathbf{I}_N \otimes \mathbf{r} \right)
	\end{split}
\end{equation}
In the following, we define $\mathbf{R} = \vect^{-1} (\mathbf{r})$ for ease of notation with $ {\rm vec^{-1}}(\mathbf{r}) \triangleq [\mathbf{r}^{(0)} \mathbf{r}^{(1)} \ldots \mathbf{r}^{(R-1)}] \in \mathbb{C}^{N_x N_z \times  R}$ so that ${\rm vec^{-1}}( \vect (\mathbf{R})) = \mathbf{R}$.

The resolution of KRPCA can be tackled via the Alternating Direction Method of Multipliers (ADMM) \cite{admm}. 
The Augmented Lagrangian associated to \eqref{eq:modified_dico_rpca} is:
\begin{equation} \label{eq:modified_dico_rpca_alm}  
	\begin{split}
		l(\mathbf{L},\mathbf{r},\mathbf{U}) =  \norm{\mathbf{L}}_* + \lambda \norm{\mathbf{R}}_{2,1}
		&+ \Re \left< \mathbf{U}, \mathbf{Y} - \mathbf{L} -  \mathbf{\Psi}( \mathbf{I}_N \otimes \mathbf{r})  \right> \\
		&+ \frac{\mu}{2} \norm{ \mathbf{Y} - \mathbf{L} - \mathbf{\Psi} (\mathbf{I}_N \otimes \mathbf{r})}_F^2
	\end{split}
\end{equation}
where $\mathbf{U}$ is the matrix of Lagrange multipliers, $\lambda$ is the sparsity regularization parameter and $\mu$ is the augmented Lagrangian penalty parameter.

The resolution for the three variables can be summarized as:
\begin{itemize}
    \item The subproblem for $\mathbf{L}$ is obtained via the soft thresholding operator on singular values, that is the proximal (see e.g. \cite{parikh2014proximal} for a comprehensive review) of the nuclear norm (with threshold $\lambda$), denoted $D_\lambda$. 
    It thus consists of the so-called $\ell_1$-norm proximal,the soft-thresholding operator $S_\lambda$, applied on the singular values of a matrix. Recall that $S_\lambda$ is defined element by element as: $[S_{\lambda}(\mathbf{A})]_{i,j}  = \sgn(a_{ij})  \left( |a_{ij}| - \lambda \right)_+$ where $\sgn(x) = x/ |x| \text{ if } x \neq 0, \text{ else }\sgn(x) = 0 \text{ for } x \in \mathbb{C}$ is the complex sign function and $(x)_+ = \max(x,0) \text{ for } x \in \mathbb{R}$. 
    Also recall the Singular Value Decomposition (SVD) denoted by $\mathbf{A} \overset{\rm{SVD}}{=} \mathbf{U} \mathbf{\Sigma} \mathbf{V}^H$, so that: 
    \begin{equation} \label{eq:svd_thres}
    \begin{split}
        D_\lambda(\mathbf{A}) =  \mathbf{U} S_\lambda(\mathbf{\Sigma}) \mathbf{V}^H
    \end{split}
    \end{equation}
    
    \item The subproblem for $\mathbf{r}$ is not solvable via a similar proximal evaluation. However, we may use proximal gradient descent (PGD) \cite{parikh2014proximal} since the objective function of this step is a sum of two convex terms with one being non-smooth. We can use a fixed step-size which is readily computed via the Hessian of the derivable part of the objective function.
    The proximal operator of  $\ell_{2,1}$-norm (with threshold $\lambda$), denoted $T_\lambda$, operating row by row over $\mathbf{A}$, is defined for the $i^{th}$ row $\mathbf{A}_{i:}$ as:
    \begin{equation} \label{eq:row_thr}
    [T_{\lambda}(\mathbf{A})]_{i:} = \left( 1- \frac{\lambda}{\norm{\mathbf{A}_{i:}}} \right)_+ \mathbf{A}_{i:}
    \end{equation}
    In fact, it is the proximal of the classical $\ell_2$-norm applied on a given row, as proximals are separable over sums.

    \item Finally, the subproblem for $\mathbf{U}$ is a standard ADMM step of dual ascent. The interesting point is that the step-size is already known: it is the parameter of the Augmented Lagrangian.
    
\end{itemize}

The method is summarized in Algorithm \ref{alg:krpca}.

\alglanguage{pseudocode}
\begin{algorithm}[!h]
    \caption{Algorithm for KRPCA}
    \label{alg:krpca}
    \begin{algorithmic}[1]
		
    \State $\text{Have: } \{\mathbf{y}_i\}_{i=1}^{N},\{\mathbf{\Psi}_i\}_{i=1}^{N}$
    \State $\text{Choose: }\lambda, \mu$

    \State$ \mathbf{Y} \triangleq [\mathbf{y}_1, \mathbf{y}_2,  \ldots \mathbf{y}_{N}]$
    \State$ \mathbf{\Psi} \triangleq [\mathbf{\Psi}_1, \mathbf{\Psi}_2,  \ldots \mathbf{\Psi}_{N}]$
    \State $\mathbf{\Psi}_{A} \triangleq [\mathbf{\Psi}_1^T  \mathbf{\Psi}_2^T  \ldots \mathbf{\Psi}_{N}^T]^T$

    \State $\mathbf{P} = {\mathbf{\Psi}_{A}}^H \mathbf{\Psi}_{A}$
    \State $ t = 1 / \lambda_{\text{max}}(\mu \mathbf{P})$

    \State $\text{Initialize: } \mathbf{L},\mathbf{R},\mathbf{U}$ 

    \vspace{2mm}
    
    \Repeat \:
    
    \State  $\mathbf{L} =  D_{1/\mu} (\mathbf{Y}- \mathbf{\Psi}  (\mathbf{I}_N \otimes \mathbf{r})+ \mu^{-1}\mathbf{U})$
    
    \State $\mathbf{\Gamma} = \mathbf{L}-\mathbf{Y}-\mu^{-1}\mathbf{U}$
    \State  $\mathbf{n} = {\mathbf{\Psi}_{A}}^H \vect{\mathbf{\Gamma}}$
    
    \Repeat :
    \State $\mathbf{R} =T_{\lambda t} ({\rm vec^{-1}} (\mathbf{r} -  t \mu(\mathbf{n} + \mathbf{P}\mathbf{r})))$
    \State $\mathbf{r} = \vect{(\mathbf{R})}$
    \Until{stopping criterion is met}

    \State $\mathbf{U}=  \mathbf{U}+ \mu (\mathbf{Y} - \mathbf{L}-  \mathbf{\Psi} (\mathbf{I}_N \otimes \mathbf{r}))$

    \Until{stopping criterion is met}
  
\end{algorithmic}
\end{algorithm}

\subsubsection{Convergence analysis}
The convergence of the algorithm is assured by the theory surrounding ADMM. In our case, both functions in the objective functions are proper closed convex functions. Assuming the (non-augmented) Lagrangian has a saddle-point, this ADMM algorithm for KRPCA guarantees residual convergence, objective convergence and dual convergence \cite{admm}. 

\subsubsection{Computational complexity}

The computational complexity of the derived algorithm for KRPCA is dependant on some assumptions on the order of the dimensions considered. 
Assume that $MN > D > M > N$ where $D=N_x N_z R$ is the  discretized scene grid size for all multipaths (and recall that $M,N$ are respectively the number of frequencies and radar snapshots).
Under those assumptions, the major cost of the overall algorithm is the computation of $\mathbf{P}$ during the initialization, which is of complexity $\mathcal{O}(MND^2)$.
In the case of repeated calls of KRPCA (e.g. for Monte Carlo simulations), we can look only at the cost of the inner loop, considering $\mathbf{P}$ as cached. 
Denoting that $\mathbf{\Psi} (\mathbf{I}_N \otimes \mathbf{r}) = \vect^{-1}(\Psi_A \mathbf{r})$, this operation can be seen to be of complexity $\mathcal{O}(MND)$. The PGD step has complexity $\mathcal{O}(KD^2)$. We assume that $K=1$ or relatively small, which is respectable in practice, otherwise the ordering of the different dimensions becomes too tight to make general statements. We then conclude that the inner loop has complexity $\mathcal{O}(MND)$. 

\section{HKRPCA : handling outliers via a robust distance in low-rank plus sparse decomposition methods}
\label{sec:robust_krpca}

The performance of KRPCA and other methods using a least squares data closeness term is susceptible to heterogeneous noise or outliers that may well appear in the context of TWRI. Indeed, as described in \cite{Ollila2012_ces_review}, most radar clutter types can be described as heterogeneous. For example, in the context of TWRI, a drywall will not have homogeneous returns in power across measurement positions. Moreover, the wall characteristics (permittivity and conductivity) may be dependant on frequency i.e. the wall is dispersive \cite{amin_TTW_book}. 

\subsection{Problem statement}

In order to alleviate the potential problems in estimation caused by heterogeneous noise or outliers, we set out to include a robust distance \cite{maronna2019robust} in our problem formulation to model the data closeness.

This leads us to define the following optimization problem, which we call HKRPCA (for Huber-type KRPCA):
\begin{equation} \label{eq:unconstr_robust_rpca}
	\min_{\mathbf{L},\mathbf{R}} \quad \norm{\mathbf{L}}_* + \lambda \norm{\mathbf{R}}_{2,1} +  \frac{\mu}{2} \sum_{p_i \in \mathcal{P}}  H_c ( \norm{[\mathbf{Y} - \mathbf{L} - \mathbf{\Psi} (\mathbf{I}_N \otimes \vect(\mathbf{R}))]_{p_i} }_F)
\end{equation}
with $ \mathcal{P}$ a \textbf{partition} of the entries of the residual matrix with $i^{th}$ element $p_i$ and $ H_c$ the renowned Huber loss function  \citet{huber1964} with threshold $c \in \mathbb{R}^+$, defined $\forall x \in \mathbb{R}$ as:
\begin{equation}
	H_c (x)  = \begin{cases} \frac{1}{2} x^2 &\mbox{if } |x| \leq c \\
		c (|x| - \frac{1}{2} c) & \mbox{if } |x| > c \end{cases}
\end{equation}
The rationale behind such a function is that outliers are higher contributors to the data closeness term than other points. Having a linear term in the loss specifically for them will lower their influence while the inliers will contribute to the loss via a quadratic term, similarly to classical least squares. 

The flexible block-wise partition of entries allows us to model the outliers shape as we see fit. 
\review{For example, if the wall materials are structured rather than homogeneous, the noise power may be variable by radar position, which induces a column-wise heterogeneity that can be taken into account in a column-wise partition.}


\subsection{ADMM algorithm with a semi-split of variables}

Handling the problem \eqref{eq:unconstr_robust_rpca} directly can be achieved by proximal gradient descent alternated on the two variables. However, a strategy to obtain closed form updates is to introduce auxiliary variables to decouple the terms of the objective function.
We introduce one auxiliary variable $\mathbf{M}=\mathbf{L}$ to decouple the nuclear norm from the Huber cost. We will see later that the split of $\mathbf{r}$ does not yield a similar proximal closed form. We consider the problem:
\begin{equation} \label{eq:semi_split_hkrpca}
\begin{split}
    \min_{\mathbf{L},\mathbf{R},\mathbf{M}} &\quad \norm{\mathbf{M}}_* + \lambda \norm{\mathbf{R}}_{2,1} +  \frac{\mu}{2} \sum_{p_i \in \mathcal{P}}  H_c ( \norm{[\mathbf{Y} - \mathbf{L} - \mathbf{\Psi} (\mathbf{I}_N \otimes \vect(\mathbf{R}))]_{p_i} }_F) \\ \text{s.t.} &\quad \mathbf{M}=\mathbf{L}
\end{split}
\end{equation}

This semi-splitting problem \eqref{eq:semi_split_hkrpca} can be tackled through the ADMM framework. The Augmented Lagrangian associated with \eqref{eq:semi_split_hkrpca} is:
\begin{equation} \label{eq:alm_slack_krpca}
\begin{split}
l(\mathbf{L},\mathbf{R},\mathbf{M},\mathbf{U}) = \norm{\mathbf{M}}_* + \lambda \norm{\mathbf{R}}_{2,1} + \Re\left< \mathbf{U}, \mathbf{M}-\mathbf{L} \right> + \frac{\nu}{2} \norm{\mathbf{M}-\mathbf{L}}^2_F  \\
+ \frac{\mu}{2} \sum_{p_i \in \mathcal{P}}  H_c ( \norm{[\mathbf{Y} - \mathbf{L} - \mathbf{\Psi} (\mathbf{I}_N \otimes \vect(\mathbf{R}))]_{p_i} }_F)
\end{split}
\end{equation}

As for KRPCA, the following subsections will detail the update of each variables for minimizing $l(\mathbf{L},\mathbf{R},\mathbf{M},\mathbf{U})$.

\subsubsection{$\mathbf{L}$-update}
For this variable, the minimization consists in finding:
\begin{equation} \label{eq:Lstep_problem}
\begin{split}
    \argmin_{\mathbf{L}} &\quad  \frac{\mu}{2} \sum_{p_i \in \mathcal{P}}  H_c ( \norm{[\mathbf{Y} - \mathbf{L} - \mathbf{\Psi} (\mathbf{I}_N \otimes \vect(\mathbf{R}))]_{p_i} }_F) +  \frac{\nu}{2} \norm{ \mathbf{M}-\mathbf{L} + \frac{1}{\nu} \mathbf{U}}^2_F
\end{split}
\end{equation}

The resulting update solving for \eqref{eq:Lstep_problem} is given in the following proposition.

\begin{proposition} \label{prop:L_step_decplg}
The solution is $\forall p_i \in \mathcal{P}$ : 
\begin{equation}
\begin{split}
	[\mathbf{L}]_{p_i}  &= \text{prox}_{  (\mu/ 2\nu)  H_c \circ \norm{\cdot}_F} \left( [\mathbf{M} + \frac{1}{\nu} \mathbf{U} -\mathbf{Y} + \mathbf{\Psi} (\mathbf{I}_N \otimes \vect(\mathbf{R})) ]_{p_i} \right) \\
	 &+[\mathbf{Y} - \mathbf{\Psi} (\mathbf{I}_N \otimes \vect(\mathbf{R})) ]_{p_i}
\end{split}
\end{equation}
with the proximal defined in the proof below (equations \eqref{eq:proxhuber} and \eqref{eq:proxcompnorm}).
\end{proposition}

\begin{proof}

The problem \eqref{eq:Lstep_problem} is separable in the blocks $\{ [\mathbf{L}]_{p_i} \}$:
\begin{equation}
\begin{split}
\min_{ \{  [\mathbf{L}]_{p_i} \} } \quad \frac{\mu}{2} &\sum_{p_i \in \mathcal{P}}  H_c ( \norm{[\mathbf{Y} - \mathbf{L} - \mathbf{\Psi} (\mathbf{I}_N \otimes \vect(\mathbf{R}))]_{p_i} }_F)  \\
+ \frac{\nu}{2} &\sum_{p_i \in \mathcal{P}}\norm{ [\mathbf{M}-\mathbf{L} + \frac{1}{\nu} \mathbf{U} ]_{p_i}  }^2_F
\end{split}
\end{equation}
By the separability property of proximals \cite{parikh2014proximal}, we can consider the proximal over each block separately:

\begin{equation} \label{eq:prox_huber_L}
	\min_{ [\mathbf{L}]_{p_i} } \quad \frac{\mu}{2}  H_c ( \norm{[\mathbf{Y} - \mathbf{L} - \mathbf{\Psi} (\mathbf{I}_N \otimes \vect(\mathbf{R}))]_{p_i} }_F) +  \frac{\nu}{2} \norm{ [\mathbf{M}-\mathbf{L} + \frac{1}{\nu} \mathbf{U} ]_{p_i}  }^2_F
\end{equation}
We then compute the proximal of $f(\mathbf{X}) = H_c(\norm{\mathbf{X}+ \mathbf{B}}_F)$ with $\mathbf{B}$ a constant term.
The proximal of the Huber function has a known form \cite{first_order_methods}:
\begin{equation} \label{eq:proxhuber}
	\text{prox}_{a H_c} (x) = \left( 1 - \frac{a}{ \max (|\frac{x}{c}|,a+1)} \right) x
\end{equation}
We can then leverage a theorem of norm composition \cite{first_order_methods} to get:
\begin{equation} \label{eq:proxcompnorm}
	\text{prox}_{a H_c \circ \norm{\cdot}_F} (\mathbf{X})  = \begin{cases}\text{prox}_{a H_c} (\norm{\mathbf{X}}_F) \cdot \frac{\mathbf{X}}{\norm{\mathbf{X}}_F} &\mbox{if } \mathbf{X} \neq \mathbf{0}\\
		\mathbf{0} & \mbox{if } \mathbf{X} = \mathbf{0} \end{cases}
\end{equation}
We finally use the translation properties of proximal operators, so that, $\forall p_i \in \mathcal{P}$, the update is: 
\begin{equation}
\begin{split}
	[\mathbf{L}]_{p_i}  &= \text{prox}_{ (\mu/ 2\nu)  H_c \circ \norm{\cdot}_F} \left( [\mathbf{M} + \frac{1}{\nu} \mathbf{U} -\mathbf{Y} + \mathbf{\Psi} (\mathbf{I}_N \otimes \vect(\mathbf{R})) ]_{p_i} \right) \\
	 &+[\mathbf{Y} - \mathbf{\Psi} (\mathbf{I}_N \otimes \vect(\mathbf{R})) ]_{p_i}
\end{split}
\end{equation}
\end{proof}
This gives a closed-form update for the $\mathbf{L}$-step. We will see later that the auxiliary variable $\mathbf{M}$ as well as the dual variable $\mathbf{U}$ also have closed-forms.

\subsubsection{$\mathbf{R}$-update: via PGD} \label{subsubsec:r_step}
The minimization problem over $\mathbf{R}$ is:
\begin{equation}
\min_{\mathbf{R}} \quad \frac{\mu}{2} \sum_{p_i \in \mathcal{P}}  H_c ( \norm{[\mathbf{Y} - \mathbf{L} - \mathbf{\Psi} (\mathbf{I}_N \otimes \vect(\mathbf{R}))]_{p_i} }_F) +  \lambda \norm{\mathbf{R}}_{2,1}
\end{equation}

It is possible to use proximal gradient descent (PGD) for the minimization over this variable. We will consider the vectorized variable $\mathbf{r}$ to compute the gradient and unvectorize the solution to apply the proximal.
At iteration $t+1$, with step-size $s$, we have : 
\begin{equation}    
\mathbf{R}_{t+1}  = T_{\lambda s} \left( \vect^{-1} \left( \mathbf{r}_t  -s  \frac{\mu}{2}\mathbf{g}_t \right) \right)
\end{equation}
where $T$ is the row thresholding operator i.e. the proximal of the $\ell_{2,1}$ norm (see Equation \eqref{eq:row_thr}) and $\mathbf{g}$ is the needed gradient of the sum of Huber functions. 

\begin{proposition}
The gradient $\mathbf{g}$ w.r.t. $\mathbf{r}$ is:
\begin{equation}
\mathbf{g} = -\sum_{p_i \in \mathcal{P}}\frac{H'_c ( \norm{[\mathbf{E}]_{p_i} }_F)}{ \norm{[\mathbf{E}]_{p_i} }_F} \left( \sum_{(j,k) \in p_i}  [\mathbf{E}]_{j,k} (\mathbf{\Psi}_{k})_{j,:}^H \right) 
\end{equation} 
where $ \mathbf{E} = \mathbf{Y} - \mathbf{L} - \mathbf{\Psi} (\mathbf{I}_N \otimes \mathbf{r})$ and $\mathbf{(\Psi}_{k})_{j,:}$ denotes the $j^{th}$ line of $\mathbf{\Psi}_{k}$.
\end{proposition}

\begin{proof}
The gradient is computed accordingly to Wirtinger calculus, since we have an objective function of complex variables. Gradient descent in this setting is achieved with:
\begin{equation}
	\mathbf{g}  = 2  \frac{d}{d \mathbf{r}^*} \left(\sum_{p_i \in \mathcal{P}}  H_c ( \norm{[\mathbf{Y} - \mathbf{L} - \mathbf{\Psi} (\mathbf{I}_N \otimes \mathbf{r})]_{p_i} }_F) \right)
\end{equation}
 Using the chain rule, we get:
\begin{equation}
	\begin{split}
		\mathbf{g} = 2  \frac{d}{d \mathbf{r}^*} \sum_{p_i \in \mathcal{P}}  H_c ( \norm{[\mathbf{E}]_{p_i} }_F) = \sum_{p_i \in \mathcal{P}}   \frac{H'_c ( \norm{[\mathbf{E}]_{p_i} }_F)}{\norm{[\mathbf{E}]_{p_i} }_F} \cdot \frac{d}{d \mathbf{r}^*} \norm{[\mathbf{E}]_{p_i} }^2_F
	\end{split}
\end{equation}
with the derivative of $H_c$  being:
\begin{equation}
	H'_c (x)  = \begin{cases} x &\mbox{if } |x| \leq c \\
		c \sgn(x) &\mbox{if } |x| > c \end{cases}
\end{equation}
where $\sgn$ denotes the sign function. Finally, we compute:
\begin{equation}
\begin{split}
    \frac{d}{d \mathbf{r}^*} \norm{[\mathbf{E}]_{p_i} }^2_F &=  \sum_{(j,k) \in p_i}   \frac{d}{d \mathbf{r}^*} \left| [\mathbf{Y}]_{j,k} - [\mathbf{L}]_{j,k} - (\mathbf{\Psi}_{k})_{j,:} \mathbf{r} \right|^2 \\
    &=  \sum_{(j,k) \in p_i}  -([\mathbf{Y}]_{j,k} - [\mathbf{L}]_{j,k} - (\mathbf{\Psi}_{k})_{j,:} \mathbf{r}) (\mathbf{\Psi}_{k})_{j,:}^H
\end{split}
\end{equation}
where $(\mathbf{\Psi}_{k})_{j,:} \mathbf{r}$ is a scalar as $\mathbf{(\Psi}_{k})_{j,:}$ denotes the $j^{th}$ line of $\mathbf{\Psi}_{k}$.
Then: 
\begin{equation}
    \mathbf{g} =- \sum_{p_i \in \mathcal{P}}\frac{H'_c ( \norm{[\mathbf{E}]_{p_i} }_F)}{ \norm{[\mathbf{E}]_{p_i} }_F} \left( \sum_{(j,k) \in p_i}  [\mathbf{E}]_{j,k} (\mathbf{\Psi}_{k})_{j,:}^H \right) 
\end{equation} 
\end{proof}
The step-size can be found by backtracking line-search (via Armijo's rule) which consists in iteratively  shrinking an initialy large step-size until sufficient decrease has been achieved. In practice,the step-size does not vary over iterations so that it can be fixed to one precomputed value (linked to the Lipschitz constant of the gradient above).

The gradient $\mathbf{g}$ may be compactly written for faster implementation:
\begin{equation}
\mathbf{g} = - \mathbf{\Psi}_g \text{bdiag} (\mathbf{e}_g) \mathbf{h}_g =  - \mathbf{\Psi}_g (\mathbf{e}_g \odot (\mathbf{h}_g \otimes \mathbf{1}))
\end{equation}
where $\mathbf{1}$ is a vector of ones and $\odot$ denotes the Hadamard product. The operator
$\text{bdiag}$ assigns a block diagonal matrix to a composite vector,
$\mathbf{\Psi}_g$ collects the dictionary vectors in the innermost sum,
$\mathbf{e}_g$ the associated residues,
and $\mathbf{h}_g$ the fraction of norms in the outermost sum.
Note that $\mathbf{e}_g \odot (\mathbf{h}_g \otimes \mathbf{1})$ is faster to compute than $\text{bdiag} (\mathbf{e}_g) \mathbf{h}_g $ as it avoids summing over the zeros of the block-diagonal matrix.

\subsubsection{$\mathbf{M}$-update}
Thanks to the variable split, the update $\mathbf{M}$ appears as a classical proximal problem with closed form solution. Indeed, after completing the squared norm, the problem of solving \eqref{eq:alm_slack_krpca} over $\mathbf{M}$ consists in finding:
\begin{equation} 
\argmin_{\mathbf{M}} \quad \norm{\mathbf{M}}_*  + \frac{\nu}{2} \norm{\mathbf{M}-\mathbf{L} + \frac{1}{\nu} \mathbf{U}}^2_F
\end{equation}
which is a proximal of the nuclear norm. Thus :
\begin{equation} 
	\mathbf{M} = D_{ 1 / \nu} (\mathbf{L} - \frac{1}{\nu} \mathbf{U} )
\end{equation}
where $D$ is the singular value thresholding operator (see Equation \eqref{eq:svd_thres}) .

\subsubsection{$\mathbf{U}$-update}
Finally, the $\mathbf{U}$ update is a standard step of ADMM, the dual ascent step: 
\begin{equation} 
\mathbf{U} = \mathbf{U} + \nu ( \mathbf{M}-\mathbf{L} )
\end{equation}

The method is summarized in Algorithm \ref{alg:hkrpca_semisplit}.

\alglanguage{pseudocode}
\begin{algorithm}[!h]
    \caption{Algorithm for HKRPCA (semi variable splitting)}
    \label{alg:hkrpca_semisplit}
    \begin{algorithmic}[1]
    \State $\text{Have: } \{\mathbf{y}_i\}_{i=1}^{N},\{\mathbf{\Psi}_i\}_{i=1}^{N}$
    \State $\text{Choose: }\lambda, \mu, \nu, \eta , c, t \text{ and } \mathcal{P}$

    \State$ \mathbf{Y} \triangleq [\mathbf{y}_1, \mathbf{y}_2,  \ldots \mathbf{y}_{N}]$
    \State$ \mathbf{\Psi} \triangleq [\mathbf{\Psi}_1, \mathbf{\Psi}_2,  \ldots \mathbf{\Psi}_{N}]$
    \State$ \mathbf{\Psi}_{A} \triangleq [\mathbf{\Psi}_1^T  \mathbf{\Psi}_2^T  \ldots \mathbf{\Psi}_{N}^T]^T$
    
    \State $\text{Initialize: } \mathbf{L},\mathbf{R},\mathbf{M},\mathbf{U}$ 
    
    \vspace{2mm}
    
    \Repeat :
    \State $[\mathbf{L}]_{p_i} = \text{prox}_{ (\mu / 2\nu)  H_c \circ \norm{\cdot}_F} \left( [\mathbf{M} + \frac{1}{\nu} \mathbf{U} -\mathbf{Y} + \mathbf{\Psi} (\mathbf{I}_N \otimes \vect(\mathbf{R})) ]_{p_i} \right) \newline 
    \hspace*{3.5em} + [\mathbf{Y} - \mathbf{\Psi} (\mathbf{I}_N \otimes \vect(\mathbf{R})) ]_{p_i}  \quad \forall p_i \in \mathcal{P}$
    
    \Repeat :
        \State $\mathbf{E}= \mathbf{Y} -  \mathbf{L} - \mathbf{\Psi} (\mathbf{I}_N \otimes \vect(\mathbf{R}))$
        
        \State $\mathbf{G} =  - \vect^{-1} \left( \sum_{p_i \in \mathcal{P}}\frac{H'_c ( \norm{[\mathbf{E}]_{p_i} }_F)}{ \norm{[\mathbf{E}]_{p_i} }_F} \left( \sum_{(j,k) \in p_i}  [\mathbf{E}]_{j,k} (\mathbf{\Psi}_{k})_{j,:}^H \right) \right)$
        
        \State $\mathbf{R} = T_{\lambda s} \left( \mathbf{R} - s \frac{\mu}{2} \mathbf{G} \right)$
   
    \Until{stopping criterion is met}
   
    \State $\mathbf{M}= D_{ 1 / \nu} (\mathbf{L} - \frac{1}{\nu} \mathbf{U})$
    \State $\mathbf{U}= \mathbf{U} + \nu (\mathbf{M}-\mathbf{L})$

    \Until{stopping criterion is met}
    
    \vspace{2mm}
    		
\end{algorithmic}
\end{algorithm}

\subsection{ADMM algorithm with full variable splitting}

The update for $\mathbf{r}$ via PGD is not the only option, we may avoid the use of an unknown step-size tuned via linesearch by splitting the variable similarly to $\mathbf{L}$ to decouple the terms in appears in. If we take this route, the formulation is:
\begin{equation} \label{eq:full_split_hkrpca}
\begin{split}
    \min_{\mathbf{L},\mathbf{R},\mathbf{M},\mathbf{S}} &\quad \norm{\mathbf{M}}_* + \lambda \norm{\mathbf{S}}_{2,1} +  \frac{\mu}{2} \sum_{p_i \in \mathcal{P}}  H_c ( \norm{[\mathbf{Y} - \mathbf{L} - \mathbf{\Psi} (\mathbf{I}_N \otimes \vect(\mathbf{R}))]_{p_i} }_F) \\ \text{s.t.} &\quad \mathbf{M}=\mathbf{L} , \hspace{1em} \mathbf{S}=\mathbf{R}
\end{split}
\end{equation}
This full variable splitting problem \eqref{eq:full_split_hkrpca} can be tackled through the ADMM framework. The Augmented Lagrangian associated with \eqref{eq:full_split_hkrpca} is:
\begin{equation}
\begin{split}
l(\mathbf{L},\mathbf{R},\mathbf{M},\mathbf{S},\mathbf{U},\mathbf{V}) &= \norm{\mathbf{M}}_* + \lambda \norm{\mathbf{S}}_{2,1}
\Re\left< \mathbf{U}, \mathbf{M}-\mathbf{L} \right> + \frac{\nu}{2} \norm{\mathbf{M}-\mathbf{L}}^2_F \\
&+ \Re\left< \mathbf{V}, \mathbf{S}-\mathbf{R} \right> + \frac{\eta}{2} \norm{\mathbf{S}-\mathbf{R}}^2_F \\
&+ \frac{\mu}{2} \sum_{p_i \in \mathcal{P}}  H_c ( \norm{[\mathbf{Y} - \mathbf{L} - \mathbf{\Psi} (\mathbf{I}_N \otimes \vect(\mathbf{R}))]_{p_i} }_F) \\
\end{split}
\end{equation}

\subsubsection{$\mathbf{L}, \mathbf{M}, \mathbf{U}$-updates}
The $\mathbf{L}$, $\mathbf{M}$ and $\mathbf{U}$ updates do not change from the semi variable splitting method. Indeed, the major difference is in the $\mathbf{r}$ update.

\subsubsection{$\mathbf{R}$-update via MM}
The objective function is in this case:
\begin{equation} \label{eq:r_step_objfct_fulldecplg}
\min_{\mathbf{R}} \quad \frac{\mu}{2} \sum_{p_i \in \mathcal{P}} H_c ( \norm{[\mathbf{Y} - \mathbf{L} - \mathbf{\Psi} (\mathbf{I}_N \otimes \vect(\mathbf{R}))]_{p_i} }_F) + \frac{\eta}{2} \norm{ \mathbf{S}-\mathbf{R} + \frac{1}{\eta} \mathbf{V}}^2_F 
\end{equation}  

Via decoupling, we cannot find a similar closed-form proximal evaluation for $\mathbf{r}$ as for $\mathbf{L}$ in Proposition \ref{prop:L_step_decplg}.
Indeed, the sum of Huber functions is not separable over $\mathbf{r}$.
Instead, we will show that the Majorization-Minimization (MM) framework \cite{reviewMM} gives us a way to solve for this subproblem iteratively.
The MM framework consists in finding a local majorizing surrogate, minimizing it and iterating those steps.

\begin{proposition}
A MM scheme can be tailored which converges to a critical point of \eqref{eq:r_step_objfct_fulldecplg}, with iteration $t+1$:
\begin{equation}
\begin{split}
\mathbf{r}_{t+1} = &\left(\frac{\mu}{2} \mathbf{\Psi}_{AW(\mathbf{r}_t)}^H \mathbf{\Psi}_{AW(\mathbf{r}_t)} + \eta \mathbf{I} \right)^{-1} \times \\
&\left( \frac{\mu}{2} \mathbf{\Psi}_{AW(\mathbf{r}_t)}^H (\vect\mathbf{Y}_{W(\mathbf{r}_t)} - \vect\mathbf{L}_{W(\mathbf{r}_t)}) + (\eta \vect\mathbf{S} + \vect \mathbf{V}) \right)
\end{split}
\end{equation}
with $\mathbf{W}$ depending on $\mathbf{r}_t$, which we drop from notations below.
We have $\mathbf{L}_W = \mathbf{W} \odot \mathbf{L}$, $\mathbf{Y}_W = \mathbf{W} \odot \mathbf{Y}$, $\mathbf{\Psi}_{AW} = \vect(\mathbf{W})\mathbf{1}^T \odot \mathbf{\Psi}_{A}$ and $\mathbf{W}$ is defined by $[\mathbf{W}]_{j,k}= w_i(\mathbf{r}_t)$ where the $(j,k)^{th}$ entry is in the $i^{th}$ patch, with $w_i^2(\mathbf{r}_t)  = 1$ if $ e_i(\mathbf{r}_t) \leq c \text{ or else } w_i^2(\mathbf{r}_t) = \frac{c}{e_i(\mathbf{r}_t)}$ where $e_i(\mathbf{r}_t) =  \norm{[\mathbf{Y} - \mathbf{L} - \mathbf{\Psi} (\mathbf{I}_N \otimes \mathbf{r}_t)]_{p_i} }_F$
\end{proposition}

\begin{proof}
Consider the vectorized variable $\mathbf{r}$ whose update we can unvectorize for $\mathbf{R}$.
The first step is to find a majorizing function of $H_c(x)$ at some point $x_t$ that we will denote $G_c (x | x_t)$. It must be equal to $H_c$ at the point $x_t$ and greater at all other points.
We can use the result from \cite[Theorem 4.5]{DELEEUW20092471} :
\begin{equation}
	G_{c} (x | x_t) =  \frac{ H'_c (x_t)}{2 x_t} (x^2 - x_t^2) + H_c (x_t)
\end{equation}

This is the sharpest quadratic majorizer. We can obtain:
\begin{equation}
	G_{c} (x | x_t)  = \begin{cases} \frac{1}{2} x^2 &\mbox{if } |x_t| \leq c \\
		\frac{1}{2} \frac{c}{|x_t|} x^2  + \frac{1}{2} c (|x_t| -c )  & \mbox{if } |x_t| > c \end{cases}
\end{equation}


Note $\forall p_i \in \mathcal{P}$ that $e_i(\mathbf{r}) =  \norm{[\mathbf{Y} - \mathbf{L} - \mathbf{\Psi} (\mathbf{I}_N \otimes \mathbf{r})]_{p_i} }_F $ and $e_i(\mathbf{r}_t)$ is the same quantity but with $\mathbf{r}_t$, the variable at the previous MM iteration.

By the definition of $G$ just above, we can write: 
\begin{equation}
\argmin_{\mathbf{r}} G_c (e_i(\mathbf{r}) | e_i(\mathbf{r}_t) ) = 
\argmin_{\mathbf{r}} \frac{1}{2} w_i^2(\mathbf{r}_t)  e_i^2 (\mathbf{r})
\end{equation}
where $w_i^2(\mathbf{r}_t)  = 1$ if $ e_i(\mathbf{r}_t) \leq c \text{ or else } w_i^2(\mathbf{r}_t) = \frac{c}{e_i(\mathbf{r}_t)} $. Also note that we can sum the majorizers over all blocks to get a global one.
Then, it follows that by adding the remaining quadratic term of the objective function, we get the following majorizer at point $\mathbf{r}_t$ to the objective function \eqref{eq:r_step_objfct_fulldecplg} :
\begin{equation}
\mathcal{G}_c (\mathbf{r} | \mathbf{r}_t) = \frac{\mu}{2} \sum_{p_i \in \mathcal{P}} G_c (e_i(\mathbf{r}) | e_i(\mathbf{r}_t) ) + \frac{\eta}{2} \norm{\mathbf{r} - (\vect{\mathbf{S}} + \frac{1}{\eta} \vect{\mathbf{V}}) }^2_F
\end{equation}
So that, via the MM framework, we are left with finding:
\begin{equation} \label{eq:mm_pblm}
\begin{split}
    \mathbf{r}_{t+1} =& \argmin_{\mathbf{r}} \mathcal{G}_c (\mathbf{r} | \mathbf{r}_t) \\
    =&\argmin_{\mathbf{r}} \quad \frac{\mu}{4} \sum_{p_i \in \mathcal{P}} w_i^2(\mathbf{r}_t)  e_i^2 (\mathbf{r}) + \frac{\eta}{2} \norm{\mathbf{r} - (\vect{\mathbf{S}} + \frac{1}{\eta} \vect{\mathbf{V}}) }^2_F  \\
    =&\argmin_{\mathbf{r}} \quad \frac{\mu}{4} \norm{\mathbf{W} \odot (\mathbf{Y} - \mathbf{L} - \mathbf{\Psi} (\mathbf{I}_N \otimes \mathbf{r}))}_F^2 \\ 
    &\hspace{4em} + \frac{\eta}{2}\norm{\mathbf{r} - (\vect{\mathbf{S}} + \frac{1}{\eta} \vect{\mathbf{V}}) }^2_F \\
\end{split}
\end{equation}
where $\mathbf{W}$ is such that $[\mathbf{W}]_{j,k}= w_i(\mathbf{r}_t) $ where the $(j,k)^{th}$ entry is in the $i^{th}$ patch.
To find the minimizer in \eqref{eq:mm_pblm}, we vectorize the first term since the Frobenius norm acts component-wise.
Then:
\begin{equation}
\begin{split}
\mathbf{r}_{t+1} &= \argmin_{\mathbf{r}} \quad  \frac{\mu}{4}
\norm{\mathbf{\Psi}_{AW} \mathbf{r} - (\vect\mathbf{Y}_W - \vect\mathbf{L}_W)}_F^2 \\
&+ \frac{\eta}{2} \norm{\mathbf{r} - (\vect{\mathbf{S}} + \frac{1}{\eta} \vect{\mathbf{V}}) }^2_F
\end{split}
\end{equation}
where $\mathbf{L}_W = \mathbf{W} \odot \mathbf{L}$, $\mathbf{Y}_W = \mathbf{W} \odot \mathbf{Y}$ and
$\mathbf{\Psi}_{AW} = \vect(\mathbf{W})\mathbf{1}^T \odot \mathbf{\Psi}_{A}$.
Via the first order optimality conditions, we get:
\begin{equation}
\begin{split}
\mathbf{r}_{t+1} = &\left(\frac{\mu}{2} \mathbf{\Psi}_{AW(\mathbf{r}_t)}^H \mathbf{\Psi}_{AW(\mathbf{r}_t)} + \eta \mathbf{I} \right)^{-1} \times \\
&\left( \frac{\mu}{2} \mathbf{\Psi}_{AW(\mathbf{r}_t)}^H (\vect\mathbf{Y}_{W(\mathbf{r}_t)} - \vect\mathbf{L}_{W(\mathbf{r}_t)}) + (\eta \vect\mathbf{S} + \vect \mathbf{V}) \right)
\end{split}
\end{equation}

\end{proof}

Finally, the $\mathbf{S}$,$\mathbf{V}$ updates are found in closed form.

\subsubsection{$\mathbf{S}$-update}
The update for $\mathbf{S}$ can be expressed as:
\begin{equation} 
\begin{split}
\min_{\mathbf{S}}  &\quad \lambda \norm{\mathbf{M}}_{2,1}  + \frac{\eta}{2} \norm{\mathbf{S}-\mathbf{R} + \frac{1}{\eta} \mathbf{V}}^2_F
\end{split}
\end{equation}
whose solution is a proximal of the $\ell_{2,1}$-norm:
\begin{equation} 
\mathbf{S} = T_{\lambda / \eta} (\mathbf{R} - \frac{1}{\eta} \mathbf{V})
\end{equation}
where  $T$ is the row thresholding operator.

\subsubsection{$\mathbf{V}$-update}
The $\mathbf{V}$-update is a generic ADMM step of dual ascent:
\begin{equation}
\mathbf{V} = \mathbf{V} + \eta ( \mathbf{S}-\mathbf{R}).
\end{equation}
Moreover, the dual balancing scheme \cite{admm} to adapt the dual hyper-parameters proved useful in practice.
The method is summarized in Algorithm \ref{alg:hkrpca_fullsplit}.

\alglanguage{pseudocode}
\begin{algorithm}[!h]
    \caption{Algorithm for HKRPCA (full variable splitting)}
    \label{alg:hkrpca_fullsplit}
    \begin{algorithmic}[1]
    \State $\text{Have: } \{\mathbf{y}_i\}_{i=1}^{N},\{\mathbf{\Psi}_i\}_{i=1}^{N}$
    \State $\text{Choose: }\lambda, \mu, \nu, \eta , c \text{ and } \mathcal{P}$

    \State$ \mathbf{Y} \triangleq [\mathbf{y}_1, \mathbf{y}_2,  \ldots \mathbf{y}_{N}]$
    \State$ \mathbf{\Psi} \triangleq [\mathbf{\Psi}_1, \mathbf{\Psi}_2,  \ldots \mathbf{\Psi}_{N}]$
    \State$ \mathbf{\Psi}_{A} \triangleq [\mathbf{\Psi}_1^T  \mathbf{\Psi}_2^T  \ldots \mathbf{\Psi}_{N}^T]^T$
    
    \State $\text{Initialize: } \mathbf{L},\mathbf{R},\mathbf{M},\mathbf{S},\mathbf{U},\mathbf{V},\mathbf{W}$ 
    
    \vspace{2mm}
    
    \Repeat :
    \State $[\mathbf{L}]_{p_i} = \text{prox}_{(\mu / 2\nu) H_c \circ \norm{\cdot}_F} \left( [\mathbf{M} + \frac{1}{\nu} \mathbf{U} -\mathbf{Y} + \mathbf{\Psi} (\mathbf{I}_N \otimes \vect(\mathbf{R})) ]_{p_i} \right) \newline 
    \hspace*{3.5em} + [\mathbf{Y} - \mathbf{\Psi} (\mathbf{I}_N \otimes \vect(\mathbf{R})) ]_{p_i}  \quad \forall p_i \in \mathcal{P}$
    
    \Repeat :
    
    \State $\mathbf{E}= \mathbf{Y} -  \mathbf{L} - \mathbf{\Psi} (\mathbf{I}_N \otimes \vect(\mathbf{R}))$
    \State $[\mathbf{W}]_{j,k}=  1 \text{ if } \norm{[\mathbf{E}]_{p_i}}_F \leq c   \text{ else }  \sqrt{c /\ \norm{[\mathbf{E}]_{p_i}}_F}  \quad \forall (j,k)\in p_i$
    \State $\mathbf{\Psi}_{AW} = \vect(\mathbf{W})\mathbf{1}^T \odot \mathbf{\Psi}_{A}$

    \State $ \mathbf{\Psi}_{AWI}  = \left(\frac{\mu}{2}\mathbf{\Psi}_{AW}^H \mathbf{\Psi}_{AW} + \eta \mathbf{I} \right)^{-1}$
    
    \State $\mathbf{r} = \mathbf{\Psi}_{AWI} \left( \frac{\mu}{2}\mathbf{\Psi}_{AW}^H (\vect\mathbf{Y}_W - \vect\mathbf{L}_W) + \frac{1}{\eta}\vect\mathbf{S} + \vect \mathbf{V} \right)$
    \Until{stopping criterion is met}
   
    \State $\mathbf{M}= D_{ 1 / \nu} (\mathbf{L} - \frac{1}{\nu} \mathbf{U})$
    \State $\mathbf{S} = T_{\lambda / \eta} (\mathbf{R} - \frac{1}{\eta} \mathbf{V})$
    
    \State $\mathbf{U}= \mathbf{U} + \nu (\mathbf{M}-\mathbf{L})$
    \State $\mathbf{V}= \mathbf{V} + \eta (\mathbf{S}-\mathbf{R})$

    \Until{stopping criterion is met}
    
    \vspace{2mm}
    		
\end{algorithmic}
\end{algorithm}

\subsection{Convergence analysis}
\subsubsection{Semi-splitting algorithm}

We consider the semi-splitting algorithm for HKRPCA, which we can write in the following equivalent formulation to \eqref{eq:semi_split_hkrpca}:
\begin{equation}
\begin{split}
    \min_{\mathbf{L},\mathbf{R},\mathbf{M}} &\quad \norm{\mathbf{M}}_* + \lambda \norm{ \mathbf{R}  }_{2,1}  + \frac{\mu}{2} \sum_{p_i \in \mathcal{P}}  H_c ( \norm{ \mathbf{S}_{p_i} (-\vect\mathbf{Y} + \vect\mathbf{L} + \mathbf{\Psi}_A \vect \mathbf{R})}_F) \\ \text{s.t.} &\quad \vect\mathbf{M} - 
    [\mathbf{I}_{MN},\mathbf{0}_{MN \times N_x N_z R}]
    \begin{bmatrix} \vect \mathbf{L} \\ \vect \mathbf{R} \end{bmatrix}
    = \mathbf{0}_{MN}
\end{split}
\end{equation}
where $\mathbf{S}_{p_i}$ denotes the selection matrix associated to the $i^{th}$ block, which has a unique or no unit entry in each column/row and zeros elsewhere. $\mathbf{0}_{M \times N} $ denotes a matrix of zeros of $M$ rows by $M$ columns. Such a zeroes matrix acts on $\mathbf{r}$, as it is not split.

Then, the above problem may be cast in a 2-block ADMM with one composite variable $[\vect(\mathbf{L})^T, \vect(\mathbf{R})^T]^T$ with coefficient matrix $[\mathbf{I}_{MN},\mathbf{0}_{MN} \mathbf{\Psi}_A ] = [\mathbf{I}_{MN},\mathbf{0}_{MN \times N_x N_z R} ]$ and associated composite convex objective function being the two latter terms of the objective function fused together.


In practice, solving directly over the composite variable is difficult so we solve for its sub-variables separately in a pass of Block Coordinate Descent (BCD), which is inexact and not part of the standard ADMM framework. Some works denoted Generalized ADMM (GADMM) \cite{fang2015_gadmm} have been developed for approximate minimization but involve the introduction of a relaxation factor which changes the problem to solve. 

We might think to cast the problem in a 3-block ADMM, which has been a topic of research the past few years \cite{Han2022ASO,chen2016}: not necessarily convergent, a simple sufficient condition for its convergence is that any two coefficient matrices in the constraints must be orthogonal to each other. But, in our case, the objective function is not separable in the different components of the composite variable, so that we cannot apply the 3-block ADMM.

Thus, to the best of our knowledge, the analysis of the convergence of such a BCD split in a 2-block ADMM remains an open question while our experiments in the following Section \ref{sec:experiments} show its good practical recovery of the seeked result. The alternative use of GADMM may be investigated but will necessitate to solve new subproblems and to verify some additional suboptimality conditions.
\subsubsection{Full-splitting algorithm}

In the case of a full split of variables i.e. splitting both $\mathbf{L}$ and $\mathbf{r}$, we can rewrite the problem in the equivalent formulation:
\begin{equation}
\begin{split}
    \min_{\mathbf{L},\mathbf{R},\mathbf{M},\mathbf{S}} &\quad \norm{\mathbf{M}}_* + \lambda \norm{\mathbf{S}}_{2,1} +  \frac{\mu}{2} \sum_{p_i \in \mathcal{P}}  H_c ( \norm{ \mathbf{S}_{p_i} (-\vect\mathbf{Y} + \vect\mathbf{L} + \mathbf{\Psi}_A \vect \mathbf{R})}_F) \\ \text{s.t.} &\quad 
    \begin{bmatrix} \vect \mathbf{M} \\ \vect \mathbf{S} \end{bmatrix}
    - \begin{bmatrix} \vect \mathbf{L} \\ \vect \mathbf{R} \end{bmatrix}
    = \begin{bmatrix} \mathbf{0}_{MN} \\ \mathbf{0}_{N_x N_z R} \end{bmatrix}
\end{split}
\end{equation}
we see that it lies in the realm of 2-block ADMM with two composite variables $[\vect(\mathbf{L})^T, \vect(\mathbf{R})^T]^T$ and $[\vect(\mathbf{M})^T, \vect(\mathbf{S})^T]^T$, with the latter term having associated objective function the sum of nuclear and $\ell_{2,1}$ norms.
 
Again, we only do inexact minimization over $[\vect(\mathbf{L})^T, \vect(\mathbf{R})^T]^T$ as well as for $[\vect(\mathbf{M})^T, \vect(\mathbf{S})^T]^T$ via Block Coordinate Descent (BCD). The question of its convergence is thus also open while experiments show good results. 

\subsection{Computational complexity}
\subsubsection{Semi-splitting algorithm}

Assuming the same ordering of dimensions as for KRPCA, i.e. that $D > M > N$ where $D=N_x N_z R$ and that $MN > D$.
The proximal of the Huber function composed with the Frobenius norm (plus a translation) is not the most costly operation as it scales linearly with the input matrix dimensions (so it is $\mathcal{O}(MN)$). 
The evaluation of $\mathbf{\Psi}(\mathbf{I}_N \otimes\mathbf{r})$ is $\mathcal{O}(MND)$ as well as for the gradient evaluation in the PGD. Setting the number of PGD iterations to $K$, we have a computational complexity of $\mathcal{O}(KMND)$ for the algorithm.

\subsubsection{Full-splitting algorithm}

Via full splitting, so with a MM step for $\mathbf{r}$, we have the task of inverting a matrix at each MM iteration (or solving the associated linear system of equations) of size $D$ which will be $\mathcal{O}(D^3)$ via Gaussian elimination. 
However, the major cost is the computation of the matrix product $\mathbf{\Psi}_{AW}^H \mathbf{\Psi}_{AW}$ inside the inverse, which will be $\mathcal{O}(NMD^2)$ and cannot be cached.
This time again, consider $K$ iterations of MM .Then, the cost of the $\mathbf{r}$-update via MM is $\mathcal{O}(KMND^2)$, which will be the overall computational complexity of the full splitting algorithm.
\review{Table \ref{tab:comput_complexities} recapitulates the complexities of all algorithms proposed in this paper. We see the higher iteration cost of the full decoupling method compared to the semi-decoupling one.}

\begin{table}[!h]
\review{
\begin{center}
\begin{tabular}{ |c|c|c|c| } 
\hline
Method &  KRPCA & HRKRPCA SD & HKRPCA FD \\
\hline
Complexity & $\mathcal{O}(MND)$ & $\mathcal{O}(KMND)$  & $\mathcal{O}(KMND^2)$  \\
\hline
\end{tabular}
\end{center}
\caption{Computational complexity of the introduced methods}
\label{tab:comput_complexities}
}
\end{table}

\review{
Figure \ref{fig:convs} presents a study of the convergence speed of the different methods.
In the point-block method, $\forall p_i \in \mathcal{P}$, $p_i$ is the support of the $i^{th}$ entry of $\mathbf{Y}$ in some chosen order. We denote this setup for the semi-decoupling algorithm as HKRPCA SD-pt and HKRPCA FD-pt for the full-decoupling algorithm.
In the column-block method, $\forall p_i \in \mathcal{P}$, $p_i$ is the support of the $i^{th}$ column $\mathbf{y}_i$. We denote this setup for the semi-decoupling algorithm as HKRPCA SD-col and HKRPCA FD-col for the full-decoupling algorithm.
It should be kept in mind that the different methods have different objective functions.
Nevertheless, we see that their convergence in terms of iterations, except SRCS, behave similarly.
Over time, we see that the point-wise HKRPCA methods (HKRPCA SD-pt and HKRPCA FD-pt) perform similarly albeit a bit slower than KRPCA, whereas their column-wise counterparts are noticeably slower (HKRPCA SD-col and HKRPCA FD-col). This is explained by the implementation: the point-wise application of the Huber function can be vectorized over the matrix, whereas the column-wise case necessitates the slicing of the matrix along the columns before applying the Huber function, which is computationally more demanding.
}

\begin{figure}[!h]
    \centering
    \includegraphics[width=0.48\linewidth]{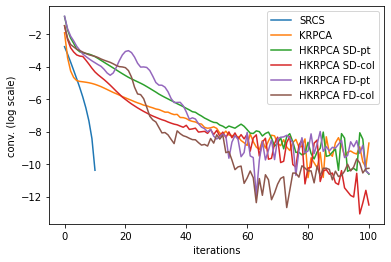}
    \includegraphics[width=0.48\linewidth]{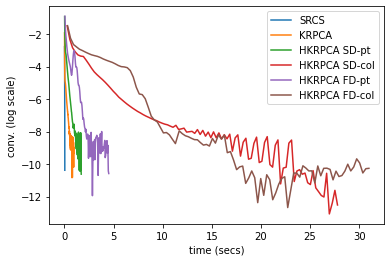}
    \caption{Convergence (log scale) vs iterations (left) and time (right)}
    \label{fig:convs}
\end{figure}

\section{Experiments} \label{sec:experiments}
\subsection{Simulation setup}
\subsubsection{FDTD data}

We test our methods on electromagnetic simulations via Finite-Difference Time-Domain (FDTD) with GprMax \cite{gprmax}.
The scene, as described in Figure \ref{fig:scene_desc}, is $4.9 \times 5.4$ m in crossrange ($x$-axis) vs downrange ($z$-axis) with a discretization step of $3$mm. The front wall, parallel to the SAR movement,is at a standoff distance to the radar of $1.2$ m. It is homogeneous and non-conductive, of thickness $20$cm  and relative permittivity $\epsilon = 4.5$. One target is behind the wall, a perfect electric conductor (PEC) cylinder of radius $3$mm situated at coordinates $(2.6,4)$. The radar moves $2$cm along the $x$-axis between each acquisition, starting from $x=1.824$m, with $67$ different positions overall. The signal sent is a ricker wavelet centered at $2$ GHz.

\subsubsection{Noise generation}
To simulate different data acquisitions, we add random heterogeneous noise drawn from student-t noise. We will consider both pointwise and column wise noise.  The column wise noise heterogeneity may arise as a result of the wall structure, e.g. drywall. The pointwise case may arise by adding the possibility of a frequency-dependant relative permittivity of the wall. Additionally, we consider the possibility of outliers coming from a different random process, which can be interpreted as mishandling in the acquisition process,etc.

We firstly consider two pointwise cases.
\begin{itemize}
\item pointwise noise only: 
$[\mathbf{Y}]_{i,j} =  [\mathbf{L} + \mathbf{\Psi} (\mathbf{I}_N \otimes \vect(\mathbf{R})]_{i,j} + [\mathbf{T}]_{i,j}$
\item pointwise noise + outliers: $[\mathbf{Y}]_{i,j} =  [\mathbf{L} + \mathbf{\Psi} (\mathbf{I}_N \otimes \vect(\mathbf{R})]_{i,j}+ [\mathbf{T}]_{i,j}  + [\mathbf{O}]_{i,j}$
\end{itemize}
with $\mathbf{T}_{i,j}$ being i.i.d. centered univariate complex t-random variables with $f > 2$ degrees of freedom (d.f.) i.e. $\mathbf{T}_{i,j} \sim \mathcal{C}t_{f}(0,\sigma)$  where the standard deviation $\sigma$ is ajusted to get the desired SNR level. 
$\mathbf{O}$ is a matrix of outliers, whose number is set by the user and whose support $\Omega$ is randomly selected at uniform among all entries. The outliers are then drawn from a standard gaussian i.e. $\mathbf{O}_\Omega \sim \mathcal{CN}(\mathbf{0},\mathbf{I})$. Entries of $\mathbf{O}$ not in $\Omega$ are then set to zero.

Secondly we consider two column wise cases.
\begin{itemize}
\item column wise noise only: 
$\mathbf{y}_i=  [\mathbf{L} + \mathbf{\Psi} (\mathbf{I}_N \otimes \vect(\mathbf{R})]_{:,i} + [\mathbf{T}]_{:,i}$
\item column wise noise + outliers: $\mathbf{y}_i=  [\mathbf{L} + \mathbf{\Psi} (\mathbf{I}_N \otimes \vect(\mathbf{R})]_{:,i} + [\mathbf{T}]_{:,i}  + [\mathbf{O}]_{:,i}$
\end{itemize}
where columns of $\mathbf{T}$ are i.i.d. random variables drawn from a $m$-variate t-distribution: $\mathbf{T}_{:,i} \sim \mathcal{C}t_{m,f}(\mathbf{0},\sigma\mathbf{I})$ with $f>2$. The outlying columns are selected uniformly at random among all columns, with their support denoted $\Omega$. The entries of $\mathbf{O}$ on those columns then follow a standard gaussian distribution i.e. $\mathbf{O}_\Omega \sim \mathcal{CN}(\mathbf{0},\mathbf{I})$ while entries not supported on $\Omega$ are set to zero.


\subsubsection{Hyperparameter tuning}
\review{
The hyperparameters have been tuned by hand in the following study.
}
For fair comparisons, all algorithms are used with hyperparameters (when applicable): $\lambda=1, \mu=10, \nu=1, c=0.1, \eta=1\text{e}10$, which have given good results for all methods. They are run the same number of iterations, as all algorithms iterations cycle through every variable, and a comparison in terms of convergence is not possible, the methods converging based on different functionals.
\review{
In order to avoid this tedious process, one may alternatively tune the hyperparameters using bayesian optimisation (see e.g. \cite{bayes_hyperparam_tuning} and references therein).
It uses a Gaussian Process (GP) prior over the f1-score of the detection map of the algorithm to tune. It is then possible to get an analytical formula for the posterior GP and to find sample hyperparameters to evaluate next based on some metric such as Expected Improvement. This can be readily implemented with the package BayesianOptimization \cite{bayesopt}.
 }

\review{
The influence of the hyperparameters $(\lambda,\mu)$ on the performance of HKRPCA has been studied in Figure \ref{fig:auc_grid}. There, each point's Area Under the Curve (AUC) is averaged over $30$ draws. 
We see that there is a fairly large range of values $\lambda \in [0,20], \mu \in [1,100]$ where the AUC is high. Additionally, we observed empirically that the Bayesian hyperparameter tuning method does propose values in this area (e.g. $\lambda = 14, \mu= 99$ here).
}

\begin{figure}[!h]
\centering
\includegraphics[width=0.48\linewidth]{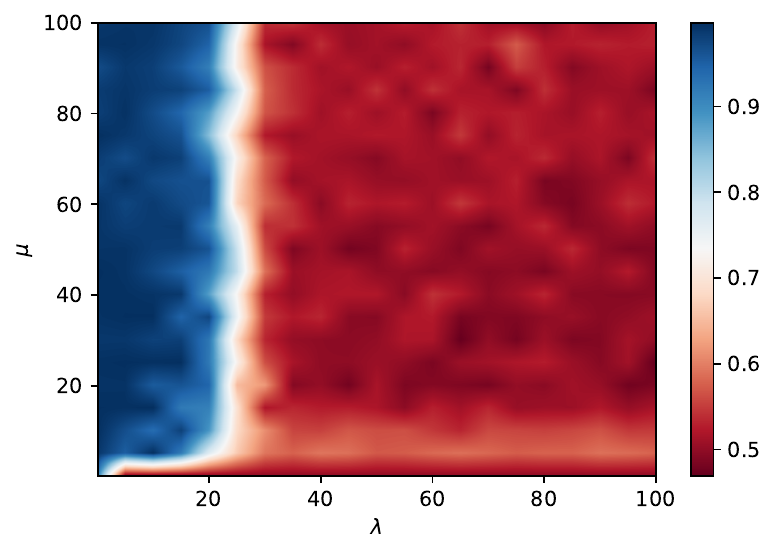}
\includegraphics[width=0.48\linewidth]{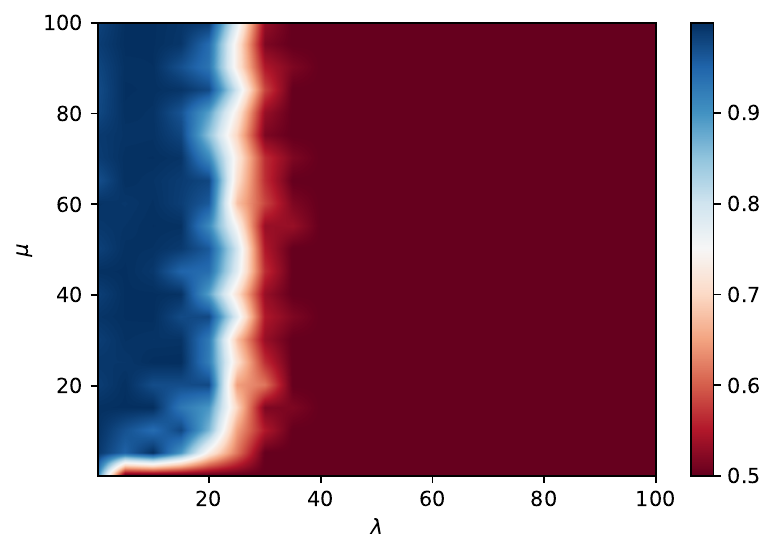}

    \caption{AUC over a grid of hyperparameters for HKRPCA FD-pt (left) and HKRPCA SD-pt (right) with pointwise noise}
    \label{fig:auc_grid}
\end{figure}

\subsection{Performance evaluation}
\label{sec:perf_eval}

Sample results are shown for the different methods in Figure \ref{fig:detectmaps} with pointwise noise only. \review{The target location is indicated with a red circle.}

We evaluate quantitatively the performance of the methods based on their Receiver Operator Characteristic (ROC) averaged over $100$ draws at each point of the curve.

\begin{figure}[!h]
\centering
\begin{subfigure}{0.32\textwidth} 
  \includegraphics[width=\textwidth]{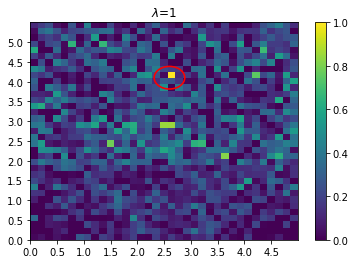}
  \caption{SRCS}
  \label{fig:srcs_detectmap}
\end{subfigure}
\begin{subfigure}{0.32\textwidth} 
  \centering
  \includegraphics[width=\linewidth]{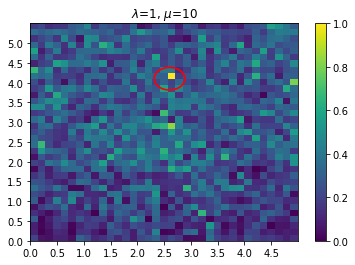}
  \caption{KRPCA}
  \label{fig:krpca_detectmap}
  \end{subfigure}
\begin{subfigure}{0.32\textwidth} 
  \centering
  \includegraphics[width=\linewidth]{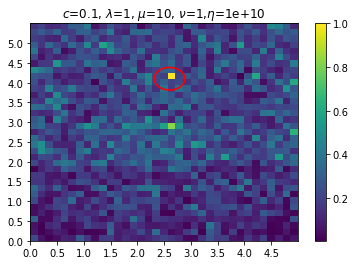}
  \caption{HKRPCA FD-pt}
  \label{fig:hkrpca_detectmap}
  \end{subfigure}
\caption{Sample detection maps (one target with location circled in red)}
 \label{fig:detectmaps}
\end{figure}

\subsubsection{Pointwise noise only}
We begin with a setup consisting in only pointwise heterogeneous noise, that is following a centered multivariate student-t distribution. We chose the setup of degrees of freedom: $ \text{d.f.}= 2.01$ and Signal to Noise Ratio: $\text{SNR}= 10\text{dB}$ to visualize at best the difference in performance of the different methods.
On Figure \ref{fig:roc_pt_noouts}, we plotted the resulting ROC.
We observe that all methods with the Huber cost perform in a similar fashion. KRPCA performs worse and finally SRCS is the worst performing method.

\subsubsection{Pointwise noise and point outliers}
Next, we are interested in a setup with pointwise heterogeneous noise plus $100$ point outliers, i.e. with perturbations coming from a different random process. Here the outlying entries have pointwise noise generated from a univariate standard gaussian distribution.
On Figure \ref{fig:roc_pt_outs}, we have the resulting ROC.
We see that both HKRPCA SD-pt and HKRPCA FD-pt perform similarly and better than  HKRPCA SD-col and HKRPCA FD-col. KRPCA and SRCS are the least well perfoming again.

\begin{figure}[!h]
\centering
\begin{subfigure}{0.49\textwidth} 
  \includegraphics[width=\textwidth]{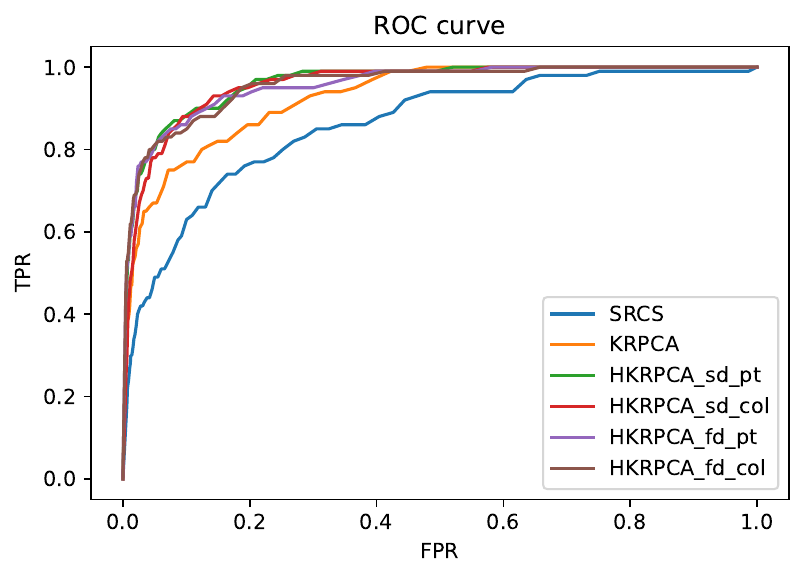}
  \caption{ROC with only pointwise heterogeneous noise (student pointwise noise with d.f. = $2.01$ and SNR= $10$ dB)}
  \label{fig:roc_pt_noouts}
\end{subfigure}
\begin{subfigure}{0.49\textwidth} 
  \centering
  \includegraphics[width=\linewidth]{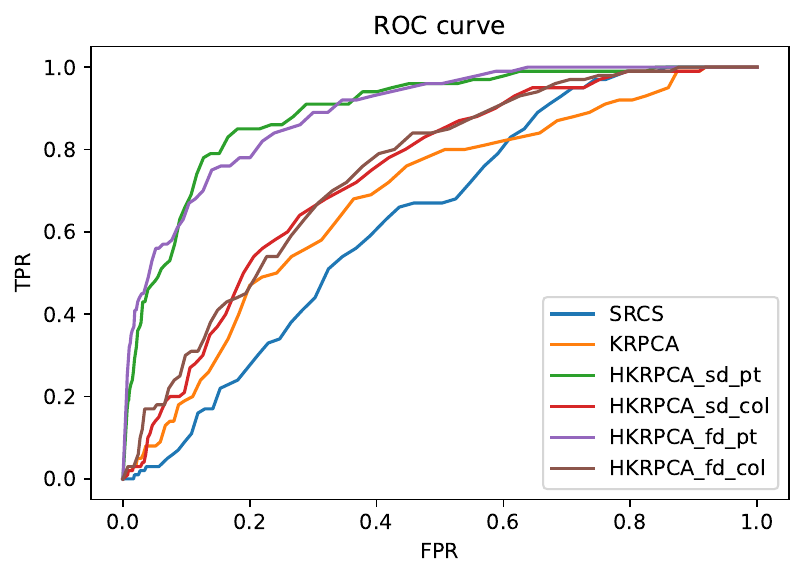}
  \caption{ROC with both pointwise noise and point outliers (student pointwise noise with d.f. = $2.1$, SNR= $12$ dB and  $100$ point outliers)}
  \label{fig:roc_pt_outs}
  \end{subfigure}
\caption{ROC with pointwise corruptions}
\end{figure}

\subsubsection{Column wise noise only}
To evaluate the effect of the block-wise methods, we thus generate block-wise noise to see its effects and the resulting discrepancy in performance of the different methods.
On Figure \ref{fig:roc_col_noouts} we have the resulting ROC with column-wise heterogeneous noise. In our setup, this means that the noise is considered radar position per radar position, and may change in power over radar acquisitions.
We see here, with a bit more degraded setup than previous ones, that HKRPCA FD-col performs the best. Other methods except SRCS are a bit below on the graph, and SRCS is last.

\subsubsection{Column wise noise and column outliers}
For one last setup, we add outliers to the column-wise setup.
To the column-wise heterogeneous noise, we add $25$ column outliers i.e. with column-wise noise generated from a standard multivariate  gaussian.
On Figure \ref{fig:roc_col_outs} we have the corresponding ROC.
We see a clearer separation of performance between all methods.
HRKPCA FD-col performs better than HKRPCA SD-col which in turns performs better than HRKPCA FD-pt. The method HRKPCA SD-pt comes after and KRPCA and SRCS are last.

On the whole, we have seen that the robust cost methods do perform better in heterogeneous noise scenarios, and that the correct block structure does impact the performance of those robust methods, especially and more clearly with outliers. Finally, the full decoupling method with a MM step performs better than the semi-decoupling method for blockwise setups.

\begin{figure}[!h]
\centering
\begin{subfigure}{0.49 \textwidth}
  \includegraphics[width= \textwidth]{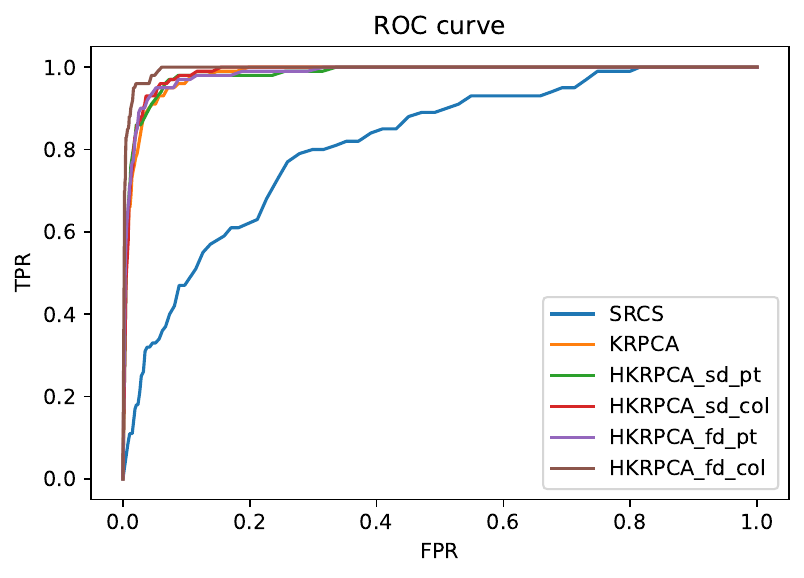}
  \caption{ROC with only column-wise heterogeneous noise (student columnwise noise with d.f. = $2.01$ and SNR= $6$ dB)}
  \label{fig:roc_col_noouts}
\end{subfigure}
\begin{subfigure}{0.49 \textwidth}
  \centering
  \includegraphics[width= \textwidth]{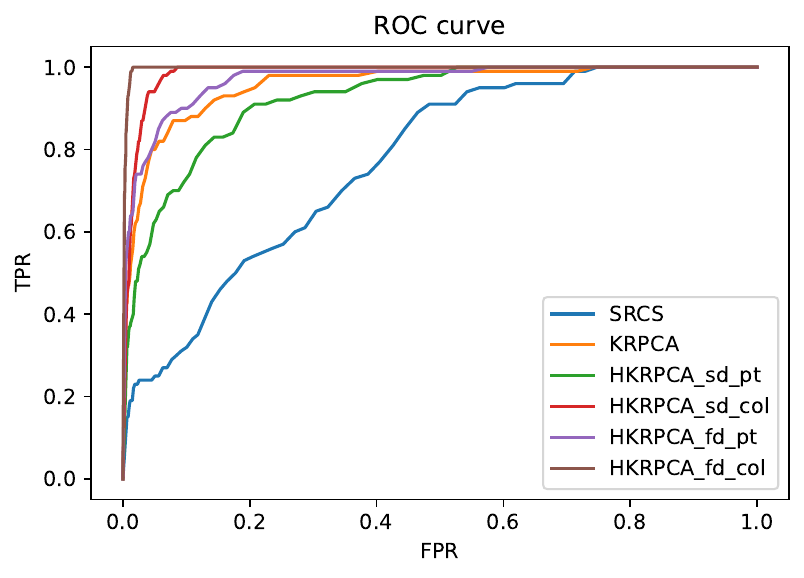}
  \caption{ROC with both column noise and column outliers (student columnwise noise with d.f. = $2.1$, SNR= $12$ dB and  $25$ column outliers)}
  \label{fig:roc_col_outs}
\end{subfigure}
\caption{ROC with column-wise corruptions}
\end{figure}

\section{Conclusion} \label{sec:conclusion}
In this paper, we stated a new method of one-step localisation of targets in the context of TWRI. It is designed to be robust to heterogeneous noise and outliers. The proposed resolution relies on the ADMM framework with two distinct algorithms tailored. One the one hand, a single split of the variable comprising the wall returns results in a closed form proximal evaluation. On the other hand, an additional split of the variable comprising the target returns lends itself to tailored MM step. We show on FDTD simulated data, in more complex scenarios where the noise is heterogeneous or outliers are present, that our method achieves better performances. This suggest further studies on real experimental data where the wall is not an idealized dielectric slab. Additionally, the methods proposed are in fact quite generic, and may be used in similar contexts such as Ground Penetrating Radar (GPR).


\footnotesize
\bibliographystyle{hkrpca} 
\bibliography{biblio}

\end{document}